\documentclass{article}

\usepackage{cite}
\usepackage{amsmath,amssymb,amsfonts,amsthm}
\usepackage{algorithmic}
\usepackage{graphicx}

\usepackage[ruled,vlined]{algorithm2e}
\usepackage{dsfont}
\usepackage{hyperref}
\usepackage{caption}
\usepackage{subcaption}
\usepackage[margin=1in]{geometry}

\newtheorem{theorem}{Theorem}
\newtheorem{claim}{Claim}

\newtheorem{proposition}{Proposition}
\newtheorem{lemma}{Lemma}
\newtheorem{definition}{Definition}

\newcommand{\cost}[2]{{\mathcal{C}_{#1}^{(#2)}}}
\DeclareMathOperator{\calE}{\mathcal{E}}
\DeclareMathOperator{\calN}{\mathcal{N}}
\DeclareMathOperator{\bbR}{\mathbb{R}}
\DeclareMathOperator{\bbE}{\mathbb{E}}
\DeclareMathOperator{\bbN}{\mathbb{N}}

\newcommand{\mean}[1]{\langle #1 \rangle}
\newcommand{\Var}{\textnormal{Var}}
\newcommand{\smeas}{\sigma_{m}}
\newcommand{\smove}{\sigma_{d}}
\newcommand{\tb}{\overline{\theta}}
\newcommand{\centre}[2]{\mean{\theta_{-#1}^{(#2)}}}
\newcommand{\one}{\mathds{1}}
\newcommand{\minus}{\textnormal{-}}

\newcommand{\thh}{\hat{\theta}}

\newcommand{\DInh}{\Delta^{\textnormal{inh}}}
\newcommand{\DRel}{\Delta^{\textnormal{rel}}}

\newcommand{\calM}{\mathcal{M}}

\newcommand{\calH}{\mathcal{H}}
\newcommand{\W}{\mathcal{W}}

\newcommand{\pa}[1]{\left( #1 \right)}

\begin{document}

\title{Distributed Alignment Processes \\with Samples of Group Average
\thanks{\copyright 2022 IEEE. Personal use of this material is permitted. Permission from IEEE must be obtained for all other uses, in any current or future media, including reprinting/republishing this material for advertising or promotional purposes, creating new collective works, for resale or redistribution to servers or lists, or reuse of any copyrighted component of this work in other works. The published article can be found at \url{https://ieeexplore.ieee.org/document/9913713}.}
\thanks{This work has received funding from the European Research Council (ERC) under the European Union's Horizon 2020 research and innovation program (grant  agreement No 648032).}
}
\author{Amos Korman, Robin Vacus}
\date{}

\maketitle

\begin{abstract}
This paper studies a stochastic alignment problem assuming that agents can sense the general tendency of the system. More specifically, we consider $n$ agents, each being associated with a real number. In each round, each agent receives a noisy measurement of the system's average value and then updates its value. This value is then perturbed by random drift. We assume that both noise and drift are Gaussian. We prove that a distributed weighted-average algorithm optimally minimizes the deviation of each agent from the average value, and for every round. Interestingly, this optimality holds even in the centralized setting, where a master agent can gather all the agents' measurements and instruct a move to each one. We find this result surprising since it can be shown that the set of measurements obtained by all agents contains strictly more information about the deviation of Agent $i$ from the average value, than the information contained in the measurements obtained by Agent $i$ alone. Although this information is relevant for Agent $i$, it is not processed by it when running a weighted-average algorithm. Finally, we also analyze the drift of the center of mass and show that no distributed algorithm can achieve drift that is as small as the one that can be achieved by the best centralized algorithm.
\end{abstract}

\section{Introduction}
\label{sec:introduction}
\subsection{Background and Motivation}

Reaching agreement, or approximate agreement, is fundamental to many distributed systems, including: computer networks, mobile sensor systems, animal groups, and neural networks \cite{vicsek_review,reynolds,simons2004many}. In the natural world, one fascinating example occurs during {\em cooperative transport} in ants, in which a group of ants join forces to carry a large or heavy food load \cite{ncomm-ofer,gelblum2016emergent}. While carrying the load, improved consensus on the direction of forces applied by the ants amounts to faster movement of the carried load, since less forces are canceling each other. Another beautiful example concerns flocking birds (or, e.g., schooling fish) that manage to maintain a cohesive direction of movement by constantly aligning their directions of flights \cite{vicsek_review}. In these contexts, as well as in multiple other contexts, such as during clock synchronization \cite{sundararaman2005clock,sivrikaya2004time}, the space in which the process occurs is continuous, measurements are noisy and the output needs to be maintained over time, while being subject to drift. 
 
In several contexts, distributed consensus algorithms can be designed by letting each agent repeatedly execute the following two steps.
\begin{enumerate}
    \item  obtain an estimate regarding the average ``opinion''  of other agents in the system. 
    \item  shift the ``opinion'' by a certain extent towards the estimated average.
\end{enumerate}
The first step assumes that agents are somehow able to sense the general behavior of the system.
Such an ability may be quite limited in many contexts, but not in all. We shall soon provide several examples for scenarios in which agents are able to measure global parameters of the system, either directly, or indirectly. However, before we list these examples, we would like to stress that the current paper is not concerned with the mechanism behind implementing such samples, and instead focuses on optimizing the second step, corresponding to the shifts that should be made after obtaining such samples.
 
\subsubsection*{Examples of global parameter samplings} We next provide several examples for settings in which agents can sample global parameters. The emphasis is on examples from the biological world. We classify these examples into three categories.
\begin{itemize}
   \item {\em Using an external entity as a relay of information.} 
In several contexts, agents may use an external entity (or a ``leader'' agent) to gather information from agents, compute a certain function of it (e.g., the average value), and then communicate this information to all others.
Indirectly, this happens, for example,  during the process of {\em cooperative transport}, in which agents physically carry a large or heavy object, as mentioned above. In this process the carried object itself serves as a medium entity through which the global tendency of the group could be sensed by each carrying agent. In the context of ants, for example, the authors of  \cite{ncomm-ofer,mccreery2014cooperative,korman2021sequential,gelblum2016emergent} assumed that each ant uses the carried object to sense the sum of the forces exerted by all carrying ants. After measuring this sum (which represents the velocity vector of the carried object), each ant decides whether, and by how much, it should align its force to this measurement. Cooperative transport following a similar approach was also studied in the context of robotics, see, e.g., \cite{wang2016multi,berman2011study}.
Other natural examples include the behavior of a swarm of flying midges \cite{gorbonos2016long}, and during firefly synchronization \cite{guerraoui2015byzantine}. In such swarms, each midge (respectively, firefly) is affected by the acoustic (respectively, visual) field produced by all others. In these examples, the environment plays the role of relaying global information to agent. 
\item
{\em Broadcast in clique networks.} 
Obtaining measurements of the whole system is also relevant for relatively packed systems, so that the communication is essentially being broadcasted over a clique. 
In the natural world, the clique abstraction can find relevance in small flocks of birds, or in other small animal groups \cite{tamm1980bird,conradt2005consensus}.
\item
{\em Well-mixed populations.}
The assumption that agents sense global parameters of the system also finds relevance in highly mobile systems. In particular, when the population is well-mixed, the pattern of communication may be captured by a random process, in which the set of neighbors is constantly being replaced by agents chosen uniformly at random.  
In such cases, sensing the current set of neighbors can serve as a fresh random sample of the global population. This perspective has recently been adopted by multiple works in theoretical computer science (although in contrast to this paper, the typical assumption in these works is that the set of opinions is discrete), see, e.g., \cite{doerr2011stabilizing,becchetti2017simple,amos-soda,feinerman2017breathe,becchetti2020consensus,ghaffari2018nearly}. For example, in {\em Majority Dynamics}, each agent pulls the opinions of a few random agents, and changes its opinion to the majority opinion among these samples \cite{doerr2011stabilizing,becchetti2020consensus,ghaffari2018nearly,becchetti2017simple}.
As such, these protocols can be viewed as being composed of two parts: (1) obtain a noisy sample of the majority opinion in the population based on the opinions of the pulled agents, and (2) align to this majority.
\end{itemize}

Informally, the {\em alignment} problem we study, considers a group of $n$ agents positioned on the real line, 
aiming to be located as close as possible to one another. 
We stress that the setting on the line does not necessarily aim to capture the physical locations of the agents, which may actually be scattered in arbitrary domains (e.g., two or three dimensional Euclidean spaces). Instead, the positions on the line aim to capture the ``opinions'' of the agents, which, in this preliminary work, are restricted to a one dimensional space. For example, in view of general alignment problems, such as flocking, the position of an agent aims to model its direction. When flocks of birds fly, the altitude axis often remains stable, and hence the flight could be viewed as two-dimensional, implying that the direction could be represented as one-dimensional. A similar observation holds in many contexts of two-dimensional navigation, for example, during {\em cooperative transport} by ants. 
In view of clock-synchronization \cite{sundararaman2005clock,sivrikaya2004time} or firefly synchronization \cite{guerraoui2015byzantine}, the position of an agent models its clock (i.e., time), which again can be viewed as a one-dimensional space  \footnote{Depending on the application, the actual domain may be bounded, or periodic. For example, when modeling directions, the domain is $[-\pi,\pi]$ and when modeling clock synchronization, the domain may be $[0,T]$ for some phase duration $T$. Since we are interested in the cases where agents are more or less aligned, approximating an interval domain with the real line is not expected to reduce the generality of our results.}. 

Initially, agents' positions are sampled from a Gaussian distribution around 0. Execution proceeds in discrete rounds, where in each round, each agent receives a noisy measurement of its current deviation from the average position of others. Then, governed by the rules of its algorithm, each agent performs a {\em move} to re-adjust its position. Subsequently, before the next round begins, the position of each agent is perturbed following random drift. The drift component may model unreliability due to external conditions (e.g., perturbation of direction by wind), or the inability of an agent to precisely adjust itself as it wishes. Both noises in measurements and random drifts are governed by Gaussian distributions. 
We are mostly interested in the following questions:
\begin{itemize}
    \item 
Which re-adjustment rule should agents adopt if their goal is to minimize expected distance of each agent from the average position? 
\item Could further communication between agents (e.g., by sharing measurements) help?
\item Which protocols minimize the drift of the center of mass?
\end{itemize}
Importantly, we assume that agents are unaware of the actual value of their current positions,  and of the realizations of the random drifts, and instead, must base their movement decisions only on noisy measurements of relative positions. This lack of global ``sense of orientation'' prevents the implementation of the trivial distributed protocol in which all agents simply move to a predetermined point, say~0. 

One trivial algorithm is the ``fully responsive'' protocol, where in each round, each agent moves all the way to its current measurement of the average position of others. 
When drift is large compared to measurement noise,
and the number of agents is large, this protocol can become  highly efficient.  However, when measurement noise is non-negligible compared to the drift, it is expected that incorporating past measurements could enhance the cohesion, even though  drift may have changed the configuration considerably. 

Perhaps the simplest class of algorithms that take such past information into account are {\em weighted-average} algorithms. By weighing the current position against the measured position in a non-trivial way, such algorithms can potentially exploit the fact that the current position implicitly encodes information from past measurements. Indeed, in a centralized setting, when a single agent aims to estimate a fixed target relying on noisy Gaussian measurements, a weighted-average algorithm is known to be optimal, in the sense that it minimizes the expected distance between the agent and the target \cite{anderson2012optimal}. However, here the setting is more complex since it is distributed, and the objective goal is to estimate (and get closer to) the average position, which is a function of the agents decisions. 

\subsection{The alignment problem}\label{sec:model}
We consider $n$ agents located on the real line $\mathbb{R}$.
Let $I = \{1,\ldots,n\}$ be the set of agents. 
We denote by $\theta_i^{(t)} \in \bbR$ the position of Agent~$i$ at round $t$, where initially, the position of each agent is chosen independently according to a normal distribution around~$0$, with variance $\sigma_0^2$, that is, for each Agent~$i$,  $\theta_i^{(0)} \sim \calN \pa{0,\sigma_0^2}.$
As mentioned, we assume that an agent is unaware of the actual value of its current position.

Execution proceeds in discrete rounds. At round $t$, each agent $i$ receives a noisy measurement of the deviation from the current average position of all other agents. Specifically, denote the average of the positions of all agents except $i$ by:
\begin{equation*}
    \centre{i}{t} = \tfrac{1}{n-1} \sum_{j\in I\setminus i} \theta_i^{(t)}.
\end{equation*}
Let $\tb_i^{(t)} = \centre{i}{t} - \theta_i^{(t)}$ denote the {\em stretch} of Agent $i$. 
At any round $t$, for every $i \in I$, a noisy measurement of the stretch of Agent~$i$ is sampled:
    \begin{equation}
       \label{eq:noise} Y_i^{(t)} = \tb_i^{(t)} + N_{m,i}^{(t)},
    \end{equation}
where $N_{m,i}^{(t)} \sim \calN(0,\smeas^2)$. In response, Agent $i$ makes a move $d\theta_i^{(t)}$ and may update its memory state (if it has any).
Finally,  the position of Agent $i$ at the next round is obtained by adding a drift:
\begin{equation} \label{eq:update_rule}
    \theta_i^{(t+1)} = \theta_i^{(t)} + d\theta_i^{(t)} + N_{d,i}^{(t)},
\end{equation}
where $N_{d,i}^{(t)} \sim \calN(0,\smove^2)$.
All random perturbations $(N_{m,i}^{(t)})_{i \in I}$ and $(N_{d,i}^{(t)})_{i \in I}$ are mutually independent, and~$\smeas,\smove > 0$.

\begin{definition}[{\em Cost}] \label{def:cost}
The {\em cost} of Agent $i$ at a given time $t$ is the absolute value of its expected stretch at that time
\footnote{Another natural cost measure is the deviation from the average position of all agents (including the agent), i.e.,
    $\cost{i}{t}' := \bbE (| \tfrac{1}{n}\sum_{j\in I} \theta_j^{(t)} - \theta_i^{(t)}|).$
These two measures are effectively equivalent. Indeed,  $\cost{i}{t}' = \tfrac{n-1}{n} \cost{i}{t}$, thus an algorithm minimizing one measure will also minimize the other.}\label{foot_COM}, i.e.,
\[ \cost{i}{t} := \bbE \pa{\left| \tb_i^{(t)} \right|}.\]
\end{definition}
Note that the cost depends on the algorithm used by $i$ and on the algorithms used by others. As these algorithms will be clear from the context, we  omit mentioning them in notations.
\begin{definition}[{\em Optimality}] \label{def:algorithm_optimality}
We say that an algorithm is {\em optimal} if, for every~$i \in I$ and every round~$t$, no algorithm can achieve a strictly smaller cost $\cost{i}{t}$.
Note that the definition refers to a very strong notion of optimality --- it implies that if Algorithm $A$ is optimal, then for any given round $t$ and for any given Agent $i$, no algorithm $B$ can achieve better performances for Agent $i$ at round $t$ than the performances guaranteed by $A$, regardless of what $B$ did in previous rounds.
\end{definition}
\begin{definition}[{\em Weighted-average algorithms}]
A {\em weighted-average} algorithm is a distributed algorithm that is characterized by a responsiveness parameter $\rho^{(t)}$ for each round $t$, indicating the weight given to the measurement at that round. Formally, an agent~$i$ following the {\em weighted-average} algorithm $\W(\rho^{(t)})$ at round $t$, sets
\begin{equation}\label{eq:responsive} d\theta_i^{(t)} = \rho^{(t)} Y_i^{(t)}. 
\end{equation}
\end{definition}

\noindent{\bf Full communication (or centralized) model.} 
When executing a weighted-average algorithm, an agent bases its decisions solely on its own measurements. 
A main question we ask is whether, and if so to what extent, can performances be improved if agents could communicate with each other to share their measurements. In order to study the impact of communication, we compare the performances of the best weighted-average algorithm to the performances of the best algorithm in the {\em full-communication model}, where agents are free to share their measurements with all other agents at no cost. In the case that agents have identities, this setting is essentially equivalent to the following centralized setting: Consider a {\em master} agent that is external to the system. The master agent receives, at any round $t$, the stretch measurements of all agents, i.e., the collection $\{Y_j^{(t)}\}_{j\in I}$, where these measurements are noisy in the same manner as described in Eq.~\eqref{eq:noise}. Analyzing these measurements at round $t$, the master agent then instructs each agent $i$ to move by a  quantity $d\theta_i^{(t)}$. After moving, the agents are subject to drift, as described in Eq.~\eqref{eq:update_rule}. Note that the master agent is unable to ``see'' the positions of the agents, and its information regarding their locations is based only on the measurements it gathers from the agents, and on the movements it instructs. 
The goal of the master agent is to minimize the maximal cost of an agent, (where the maximum is over the agents), per round. In particular, an algorithm is said to be 
optimal in the centralized setting if it satisfies Definition \ref{def:algorithm_optimality}, (where optimality is with respect to all algorithms in the centralized setting).

\subsection{Our results}

We prove that a distributed \emph{weighted-average} algorithm, termed $\W^\star$, optimally minimizes the expected stretch of all agents at all rounds, in the strong sense of Definition~\ref{def:algorithm_optimality}.
The optimality of this weighted-average algorithm holds even in the centralized (full-communication) model. We find this result surprising since it can be shown that the measurements made by Agent $i$ contain strictly less information about the agent's relative position to the average position  than the information contained in the collection of all measurements $\{Y_j^{(t)}\}_{j\in I}$.
Indeed, although the $\{Y_j^{(t)}\}_{j\in I}$ measurements are not centered around the stretch $\tb_i^{(t)}$ of Agent~$i$, they still contain useful information about the stretch of Agent $i$. For example, it can be shown that
\begin{equation}\label{eq:moreknowledge}
    - \sum_{j \in I\setminus i} Y_j^{(t)} = \tb_i^{(t)} - \sum_{j \in I\setminus i}^n N_{m,j}^{(t)},
\end{equation}
thus representing an additional ``fresh'' estimation of $\tb_i^{(t)}$. Hence, the measurements obtained by Agent~$i$ are not sufficient statistics with respect to all measurements, when aiming to estimate the stretch of Agent~$i$.
Therefore, at a first glance, it may appear that in the centralized setting the stretch of agents could potentially be reduced by letting the master agent process all  measurements.
Nevertheless, when operating under $\W^\star$, Agent~$i$ processes merely its own measurements, and still its stretch is minimized, as it could best be minimized by utilizing all measurements. The reason why this happens, is that other agents manage to fully process the information not processed by Agent $i$ in a way that also benefits Agent~$i$, by shifting the center of mass towards it.

Finally, we also analyze the drift of the center of mass and show that no distributed algorithm can achieve as small of a drift as can be achieved by the best centralized algorithm. In light of this, we also show that the drift associated with $\W^\star$ incurs a relatively small overhead over the best possible drift in the centralized setting.
 \\

\noindent{\bf Intuition: The case of Two Agents in One round.} \label{sec:intuition}
\newcommand{\cent}[1]{{#1}^\textnormal{cent}}
\newcommand{\dist}[1]{{#1}^\textnormal{dist}}
In order to obtain preliminary understanding, let us first examine the case of two agents operating in a single round, and compare the situation between the centralized and the distributed setting. This simple example is informative since it already demonstrates the aforementioned surprising phenomenon. 

To really limit this example to ``one round'', we assume that the initial distribution of the stretches has a very large variance (formally, $\sigma_0^2 \gg \smeas^2$). Therefore, we can put aside the Bayesian conflict between the a priory knowledge about the stretches, and the information contained in the new measurement -- and only work with the latter.

When considering two agents, the {\em stretch} $\tb_i$ of Agent $i\in\{1,2\}$ is simply its relative distance to the other agent. Assume that their positions are initially set to~$\theta_1 = 0$ and $\theta_2 = 1$, so that~$\tb_1 = +1$ and $\tb_2 = -1$.
Assume further that the respective measurement noises are~$N_{m,1} = 0.1$ and $N_{m,2} = 0.2$.
In this case, the measurements of the stretches are:
\begin{equation}\label{eq:cent_measure}
    \begin{cases}
        Y_1 = \tb_1 + N_{m,1} = +1.1 \\ Y_2 = \tb_2 + N_{m,2} = -0.8.
    \end{cases}
\end{equation}

\noindent{\em \underline{Centralized setting.}} Here, a central entity gathers both measurements $Y_1$ and $Y_2$, and, based on these measurements, instructs each agent of its  move.
Since we assumed that the initial distribution of the stretches has a very large variance compared to~$\smeas^2$, the most accurate estimates of the relative stretches between the agents {depends only on the measurements} 
\begin{equation}\label{eq:est_cent}
    \begin{cases}
        \cent{~\thh_1} = \tfrac{1}{2} (Y_1 - Y_2) = \tfrac{+1.9}{2} = +0.95 \\
        \cent{~\thh_2} = \tfrac{1}{2} (Y_2 - Y_1) = \tfrac{-1.9}{2} = -0.95.
    \end{cases}
\end{equation}

Recall that the goal of the central entity is to minimize the absolute stretch between the agents. Note that for this purpose, it has a degree of freedom, since, once it moves each of them, shifting both by the same amount would not change their stretch. However, if the central entity also wishes to minimize the drift of the center of mass, then it should instruct the agents to move towards the estimated center, shifting them by~$\thh_i/2$ towards each-other. (This type of algorithm will later be called ``meet at the center''.)
Specifically, the instructed moves would be~$\cent{d\theta_1} = \tfrac{1}{2} \cent{~\thh_1} =  +0.475$ and~$\cent{d\theta_2} = \tfrac{1}{2} \cent{~\thh_2} = -0.475$.
By definition, the new positions of the agents (before drift applies) would be:
\begin{equation}\label{eq:cent_positions}
    \begin{cases}
        \cent{\theta_1'} = \theta_1 + \cent{d\theta_1} = 0.475 \\ \cent{\theta_2'} = \theta_2 + \cent{d\theta_2} = 0.525
    \end{cases} \text{so}
    \begin{cases}
        \cent{\tb_1'} = +0.05, \\ \cent{\tb_2'} = -0.05.
    \end{cases}
\end{equation}
As can be seen, the center of mass is unchanged. In fact, this is no coincidence and would be the case for any measurement noises (before drift takes place). 

\noindent{\em \underline{Distributed setting.}}
In the \textit{distributed} setting, an agent must base its estimate of the stretch on its measurement alone. In this case, we have:
\begin{equation*}
    \begin{cases}
        \dist{~\thh_1} = Y_1 = +1.1, \\ \dist{~\thh_2} = Y_2 = -0.8.
    \end{cases}
\end{equation*}
Of course, the quality of these estimates is less than in the centralized setting (compare with Eq.~\eqref{eq:est_cent}). 
Once again, in attempting to ``meet at the center'', the best that the agents can do is to go half way to their respective estimate of the location of the other agent. That is, $\dist{d\theta_1} = \tfrac{1}{2} \dist{~\thh_1} = \tfrac{1}{2} \cdot Y_1 = +0.55$ and~$\dist{d\theta_2} = \tfrac{1}{2} \dist{~\thh_2} = \tfrac{1}{2} \cdot Y_2 = -0.4$.
By definition, their new locations (before drift applies) would be:
\begin{equation*}
    \begin{cases}
        \dist{\theta_1'} = \theta_1 + \dist{d\theta_1} = 0.55 \\ \dist{\theta_2'} = \theta_2 + \dist{d\theta_2} = 0.6
    \end{cases},
    \text{ yielding}
    \begin{cases}
        \dist{\tb_1'} = +0.05 \\ \dist{\tb_2'} = -0.05.
    \end{cases}
\end{equation*}
Evidently, in contrast to the centralized case (Eq.~\eqref{eq:cent_positions}), here the center of mass moved as a result of the agents movements. Importantly, however, the relative positions of the agents (and therefore, their costs) are identical to the ones obtained in the centralized protocol. Again, this is no coincidence, and holds for any measurement noises. The reason is, that the two locations of the agents in the distributed setting differ by the same shift, compared to their location in the centralized setting. Indeed, we have:
\begin{equation*}
    \begin{cases}
     \dist{\theta_1'}- \cent{\theta_1'}  =   \dist{d\theta_1} - \cent{d\theta_1} = \tfrac{1}{4}(Y_1+Y_2) \\
        \dist{\theta_2'} - \cent{\theta_2'} =  \dist{d\theta_2} - \cent{d\theta_2}=  \tfrac{1}{4}(Y_1+Y_2).
    \end{cases}
\end{equation*}
This shift, i.e., $\tfrac{1}{4}(Y_1+Y_2)$, represents the ``mistake'' made by distributed agents in attempting to meet at the center, compared with the centralized algorithm.
However, since both agents in the distributed setting are shifted by the same quantity compared to their respective positions in the centralized setting, their relative stretch remains the same in both settings. A main challenge of this paper is to show that this phenomenon holds also in the multi-round scenario, and with arbitrary number of agents.

\noindent{\bf Weighted-average algorithms with fixed responsiveness.}
As a warm-up towards understanding multi-round scenarios with multiple agents, we first investigate weighted-average algorithms in which all agents have the same responsiveness~$\rho$, that furthermore remains fixed throughout the execution (see Eq.~\eqref{eq:responsive}). The proofs of both Theorems~\ref{thm:weighted_average1} and~\ref{thm:weighted_average2} below are given in Section~\ref{sec:warm-up_proofs}.

\begin{theorem}\label{thm:weighted_average1}
Assume that all agents execute $\W(\rho)$, for a fixed $0 \leq \rho \leq 1$. Then for every $i \in I$ and every $t \in \mathbb{N}$, the stretch $\tb_i^{(t)}$ is normally distributed, and
 \[\lim_{t \rightarrow +\infty} \Var\pa{\tb_i^{(t)}} = \frac{ \tfrac{n}{n-1}(\rho^2\smeas^2 + \smove^2)}{1-(1-\tfrac{n}{n-1} \rho)^2},\]
with the convention that~$\lim_{t \rightarrow +\infty} \Var({\tb_i^{(t)}}) = +\infty$ if the denominator~$1-(1-\tfrac{n}{n-1} \rho)^2 = 0$.
\end{theorem}

If all agents run $\W(\rho)$, then Agent $i$'s cost at steady state
is captured by $\text{Var}(\rho):=\lim_{t \rightarrow +\infty} \Var\pa{\tb_i^{(t)}}$. Indeed, for every $i$,  since $\tb_i^{(t)}$ is normally distributed,
$ \lim_{t \rightarrow +\infty} \mathbb{E}\left(\left|\tb_i^{(t)}\right|\right) = \sqrt{\tfrac{2}{\pi}\text{Var}(\rho)}$.
The minimal value of this
is achieved when taking $\text{argmin}_\rho \text{Var}(\rho)$ as the responsiveness parameter.

\begin{theorem}\label{thm:weighted_average2}
    The weighted-average algorithm that optimizes group variance among all  weighted-average algorithms $\W(\rho)$ (that use the same responsiveness parameter $\rho$ at all rounds) is
    $\W(\rho^\star)$, where
    \begin{equation} \label{eq:rhostar}
         \rho^\star = \frac{ \smove \sqrt{4\smeas^2 + \pa{\tfrac{n}{n-1}\smove}^2} - \tfrac{n}{n-1}\smove^2}{2\smeas^2}.
    \end{equation} 
\end{theorem}
When~$n$ is large, Eq.~\eqref{eq:rhostar} becomes $\rho^\star \approx \tfrac{\smove \sqrt{4\smeas^2 + \smove^2} - \smove^2}{2\smeas^2}$.
Note that the role played by the measurement noise is very different than the role played by the drift. 
For example, if $\smeas\gg \smove$, then $\rho^\star\approx 0$. However, if $\smeas\ll \smove$ then $\rho^\star\approx 1$. Interestingly, if $\smeas=\smove$ then $\rho^\star \approx \tfrac{\sqrt{5}-1}{2}$, which is highly related to the golden ratio.
Moreover, for large~$n$, the minimal $\Var(\rho)$ is
\begin{equation} \label{eq:varstar}
    \Var(\rho^\star) = \tfrac{1}{2}\smove\left( \sqrt{4\smeas^2 + \smove^2} + \smove \right).
\end{equation}
Note that when the measurements are perfect, i.e., $\smeas = 0$, we have $\Var(\rho^\star) = \sigma^2_d$, which is the best achievable value that an agent can hope for, since no algorithm can overcome the drift-noise. \\

\noindent{\bf The impact of communication.}
Our next goal is to understand whether, and if so, to what extent, can performances be improved if 
further communication between agents is allowed. For this purpose, we compare 
the performances of $\W(\rho^\star)$ to the performances of the best algorithm in the centralized (full-communication) setting. 

A natural candidate for an optimal algorithm in the centralized setting is the ``meet at the center'' algorithm. This algorithm first obtains, for each agent, the best possible estimate of the distance from the agent's position to the average position~$\mean{\theta}$, based on all measurements, and then instructs the agent to move by this quantity (towards the estimated average position).
However, it is not immediate to figure out the distances to the average position, and furthermore, quantify the performances of this algorithm.
To this end, we make use of the Kalman filter tool \cite{anderson2012optimal}, by adapting it to our setting. Solving the Kalman filter system associated with the centralized version of our  alignment problem, we obtain an estimate of the relative distance of each agent $i$ from the average position (based on all measurements). To describe these estimates we first define the following.

\begin{definition}[Uncertainty] \label{def:alpha}
    We inductively define the sequence $(\alpha_t)_{t=0}^\infty$. Let $\alpha_0 = n\sigma_0^2/(n-1)$, and for every~$t$, let
    \begin{equation} \label{eq:alpha}
        \alpha_{t+1} = \frac{\smeas^2 \alpha_t}{\tfrac{n}{n-1} \alpha_t + \smeas^2} + \tfrac{n}{n-1}\smove^2. 
    \end{equation} 
\end{definition} 
The value of $\alpha_t$ aims to represent the variance of the (random) error made when estimating the stretch of the agents at round~$t$. When thinking of $n$ as very large, the formula for $\alpha_{t+1}$ is reminiscent of the resistance formula in electrical circuits, where we add in parallel to the current resistor with resistance $\alpha_{t}$ a resistor with resistance $\sigma^2_m$ (standing for the measurement noise), and another resistor in a series, with resistance $\sigma^2_d$ (standing for the drift).
\begin{definition}[Optimal weight] \label{def:rho_star_t}
    For every integer~$t$, let
    \begin{equation*}
        \rho_\star^{(t)} = \frac{\alpha_t}{\tfrac{n}{n-1} \alpha_t + \smeas^2}.
    \end{equation*} 
\end{definition}
The value of $\rho_\star^{(t)}$ aims to indicate by how much to weigh the measurements of round~$t$  in order to move optimally: if~$\smeas^2 \ll \alpha_t$, then~$\rho_\star^{(t)} \approx 1$, meaning that the agent should entirely rely on the measurement. Conversely, if~$\smeas^2 \gg \alpha_t$, then~$\rho_\star^{(t)} \approx 0$, and the agent should not take the measurement into account.

At each round $t$, the Kalman filter returns an estimate of the relative distance of each agent $i$ from the average position, which turns out to be
 \[\rho_\star^{(t)} \Big( \tfrac{n-1}{n} Y_i^{(t)} - \tfrac{1}{n}\sum_{j \in I\setminus i} Y_j^{(t)} \Big).\]
As guaranteed by the properties of the Kalman filter, these estimates minimize the expected sum of square-errors, which can be translated to our desired measure of minimizing the agents' costs.
The ``Meet-at-the-center'' algorithm is given by Algorithm \ref{alg:MatC} below, and the following theorem stating its optimality is proved in Section \ref{sec:meet}.

\begin{algorithm}[ht] 
    \SetAlgoLined
    \ForEach{round~$t$}{
        Consider all measurements at round~$t$,~$\{Y_j^{(t)} \mid 1 \leq j \leq n\}$ \;
        \ForEach{agent~$i$}{
            Set $d\theta_i^{(t)} = \rho_\star^{(t)} \pa{\tfrac{n-1}{n} Y_i^{(t)} - \tfrac{1}{n}\sum_{j \neq i} Y_j^{(t)}}$ \tcc*{Output an estimate of $\mean{\theta} - \theta_i$}
        }
    }
\caption{Meet-at-the-center} \label{alg:MatC}
\end{algorithm}
\begin{theorem} \label{thm:matc_opt}
    Meet-at-the-center is optimal in the centralized setting.
\end{theorem}

Note that in the centralized setting, once we have an optimal algorithm $A$, we can derive another optimal algorithm $B$ by simply shifting all agents, at each round $t$, by a fixed quantity $\lambda_t$. Indeed, such shifts do not influence the relative positions between the agents. Conversely, we show in Appendix~\ref{sec:shifts} 
that all (optimal) deterministic algorithms in the centralized setting are  shifts of one another (note that the shifts $\lambda_t$ are not necessarily the same for all rounds $t$).

As it happens, another solution that follows from the Kalman filter estimations can be described as a distributed weighted-average algorithm  (see Algorithm \ref{alg:weighted-average}). As such, this algorithm, henceforth called $\W^{\star}$, also yields optimal stretches. 
The proof of the following theorem is given in Section \ref{sec:optW}.
\begin{theorem} \label{thm:weight_opt}
    Algorithm $\W^\star=\W({\rho_{\star}^{(t)}})$ is optimal in the centralized setting.
\end{theorem}
\begin{algorithm}[ht] 
    \SetAlgoLined
    \ForEach{round~$t$}{
        Consider all measurements at round~$t$,~$\{Y_j^{(t)} \mid 1 \leq j \leq n\}$ \;
        \ForEach{agent~$i$}{
            Set $d\theta_i^{(t)} = \rho_\star^{(t)} Y_i^{(t)}$ \;
        }
    }
\caption{Algorithm $\W^{\star}$} \label{alg:weighted-average}
\end{algorithm}
Since both Meet-at-the-center and $\W^{\star}$ are optimal deterministic algorithms they must be shifts of one another.
Indeed, as validated in \cite{korman:hal-03124213} (Appendix~D.2), by adding the shift $\lambda_t = \tfrac{1}{n} \rho_\star^{(t)}\sum_{i\in I} Y_i^{(t)}$ to the agents in Meet-at-the-center we obtain Algorithm $\W^{\star}$. 

Note that in contrast to $\W(\rho^\star)$, Algorithm $\W^{\star}$ uses a different responsiveness $\rho_{\star}^{(t)}$ at each round $t$. We next argue that the sequences $(\alpha_t)$ and~$(\rho_\star^{(t)})$ converge. Not surprisingly, at the limit, we recover the optimal responsiveness as stated in Theorem~\ref{thm:weighted_average2}. Due to lack of space the proof of the following claim is deferred to Appendix~B.2 of \cite{korman:hal-03124213}.

\begin{claim} \label{claim:asymptotic_alpha}
    The sequence $\alpha_t$ converges to
    \begin{equation*}
        \alpha_{\infty} := \tfrac{1}{2} \left( \smove \sqrt{4\smeas^2 + \pa{\tfrac{n}{n-1}\smove}^2} + \tfrac{n}{n-1}\smove^2 \right).
    \end{equation*}
    Moreover,
    $\lim_{t \rightarrow +\infty} \rho_\star^{(t)} = \rho_\star.$ 
\end{claim}

\noindent{\bf Drift of the center of mass.}
In this section we analyze the drift of the center of mass. Essentially, we show that 
with respect to this drift (and in contrast to the measure of stretch as used in Definition \ref{def:cost}), there is a gap between the centralized and the distributed settings. Specifically, we first show that Meet-at-the-center is ``drift-optimal'', in sense that it achieves the best possible drift among all algorithms in the centralized setting. In contrast, we prove that there is no distributed algorithm that can achieve as small of a drift as the one obtained by the ``meet at the center'' algorithm. Consistent with this result, Algorithm $\W^{\star}$ is not ``drift-optimal'', yet we show that the overhead it incurs in terms of this drift is relatively small, especially when $n$ is large compared to $t$ or when $\smeas$  is small compared to $\smove$.  

Let us now describe these results in more detail. Denote~$\Delta_t = \mean{\theta^{(t)}} - \mean{\theta^{(0)}}$.
By definition (Eq.~\eqref{eq:update_rule}), for a given algorithm, the center of mass changes in round $t$ as follows:
\begin{equation} \label{eq:drift_center_of_mass_t}
    \mean{\theta^{(t+1)}} - \mean{\theta^{(t)}}
    = \tfrac{1}{n} \sum_{i\in I} \pa{d\theta_i^{(t)} + N_{d,i}^{(t)}}.
\end{equation}
The drift of the center of mass after $t$ rounds is therefore 
\begin{align*} 
    \Delta_t &= \tfrac{1}{n} \sum_{s=0}^t \sum_{i\in I} \pa{d\theta_i^{(s)} + N_{d,i}^{(s)}}.
\end{align*}
Let us define {\em inherent-drift} as~$\DInh_t := \tfrac{1}{n} \sum_{s=0}^t \sum_{i\in I}  N_{d,i}^{(s)}$ and {\em relative-drift} as $\DRel_t := \tfrac{1}{n} \sum_{s=0}^t \sum_{i\in I}  d\theta_i^{(s)}$.
With these notations, we have:
\begin{equation} \label{eq:drift_center_of_mass}
    \Delta_t = \DRel_t + \DInh_t.
\end{equation}
Due to lack of space, the proof of the following lemma appears in Appendix E of \cite{korman:hal-03124213}.

\begin{lemma} \label{lem:drift_independence}
    For every protocol and for every round~$t$, $\DInh_t$ and~$\DRel_t$ are independent.
\end{lemma}

A direct consequence of Lemma~\ref{lem:drift_independence} is that
\begin{equation} \label{eq:lb_var}
    \bbE \pa{ \Delta_t^2 } = \bbE \pa{ {\DRel_t}^2 } + \Var \pa{\DInh_t} \geq \Var(\DInh_t),
\end{equation}
justifying the choice of the term ``inherent''. Since $\Var(\DInh_t)=\Var \pa{ \tfrac{1}{n} \sum_{s=0}^t \sum_{i\in I} N_{d,i}^{(s)} } = \tfrac{t+1}{n} \smove^2$, we obtain the following.

\begin{theorem}\label{thm:lower_var}
    For every protocol in the centralized setting and for every round~$t$, the drift of the center of mass satisfies $\bbE \pa{ \Delta_t^2 } \geq  \tfrac{t+1}{n} \smove^2$.
\end{theorem}

\begin{definition}\label{def:drift_optimality}
    An algorithm is called {\em drift-optimal}, if for every round $t$, the drift of the center of mass satisfies $\bbE \pa{ \Delta_t^2 } =  \Var(\DInh_t)= \tfrac{t+1}{n} \smove^2$.
\end{definition}

\begin{theorem} \label{thm:drift_matc}
    Meet-at-the-center is drift-optimal.
\end{theorem}
\begin{proof}
    When the agents follow Algorithm Meet-at-the-center, for every~$t \in \bbN$
    \begin{equation*} 
        \tfrac{1}{n} \sum_{i\in I} d\theta_i^{(t)} = \tfrac{1}{n} \sum_{i\in I} \rho_\star^{(t)} \pa{\tfrac{n-1}{n} Y_i^{(t)} - \tfrac{1}{n}\sum_{j \in I\setminus i} Y_j^{(t)}}         = \tfrac{\rho_\star^{(t)}}{n}  \pa{ \tfrac{n-1}{n}\sum_{i\in I} Y_i^{(t)} - \tfrac{1}{n} \sum_{i\in I} (n-1)\cdot Y_i^{(t)}} 
        = 0.
    \end{equation*}
    Hence, by definition, $\DRel_t=0$ and by Eq.~\eqref{eq:lb_var}, the drift of the center of mass in round~$t$ is limited to~$\DInh_t$, which concludes the proof of Theorem~\ref{thm:drift_matc}.
\end{proof}

Next, we analyze the drift of $\W^{\star}$.

\begin{theorem} \label{thm:drift_wstar}
    When Algorithm~$\W^{\star}$ is used, the drift of the center of mass satisfies $ \bbE\pa{\Delta_t^2} = \tfrac{t+1}{n} \smove^2 + \tfrac{1}{n} \sum_{s=0}^t {\rho_\star^{(s)}}^2 \smeas^2$.
\end{theorem}
\begin{proof}
    Under Algorithm $\W^{\star}$, we have for every~$t \in \bbN$
    \begin{align*}
        \tfrac{1}{n} \sum_{i\in I} d\theta_i^{(t)}
        = \tfrac{1}{n} \sum_{i\in I} \rho_\star^{(t)} \pa{ \tb_i^{(t)} + N_{m,i}^{(t)} } 
        = \tfrac{1}{n} \sum_{i\in I} \rho_\star^{(t)} N_{m,i}^{(t)},
    \end{align*}
    where the last equality is by Lemma~\ref{lem:tb_sum}. Using Eq.~\eqref{eq:lb_var}, and summing over the rounds $s=0,\ldots, t$,
    we conclude the proof of Theorem~\ref{thm:drift_wstar}.
\end{proof}

Note that since $\rho_\star^{(s)}\leq 1$, the overhead incurred by Algorithm~$\W^{\star}$ is at most $\tfrac{t+1}{n} \smeas^2$.

Finally, we argue that the sub-optimality of Algorithm $\W^\star$ (in terms of the drift of the center of mass) is not coincidental. The proof of the following claim is deferred to Appendix~\ref{sec:last_proof}.

\begin{claim} \label{claim:best_drift}
    No algorithm in the distributed setting can be simultaneously optimal (Definition~\ref{def:algorithm_optimality}) and drift-optimal (Definition~\ref{def:drift_optimality}). 
\end{claim}

\subsection{Related works}

The Kalman filter algorithm is a prolific tool commonly used in control theory and statistics, with numerous technological applications \cite{anderson2012optimal}. 
In this paper, we used the Kalman filter to investigate the centralized setting, where all relative measurements are gathered at one master agent that processes them and instructs agents how to move. Fusing relative measurements of multiple agents is often referred to in the literature as {\em distributed Kalman filter}, see surveys in \cite{cao2012overview,olfati2007distributed,ren2005survey}. However, there, the typical setting is that each agent invokes a separate Kalman filter to process its measurements, and then the resulted outputs are integrated using communication. Moreover, works in that domain often assume that observations are attained by static observers (i.e., the agents are residing on a fixed underlying graph) and that the measured target is external to the system. In contrast, here we consider a self-organizing system, with mobile agents that measure a target (average position) that is a function of their positions. Closer works to our setting where a group of mobile agents aim to gather together are \cite{tang2015continuous,huang2009coordination}. However, among the several differences with our paper, these papers assume that there is no drift and that agents communicate their absolute position, rather than their relative position. 

Similarly to our findings, the authors of~\cite{fisher} also emphasize the effectiveness of performing a weighted-average between the opinion of the individual and the sample it receives.
However, we find this result less surprising under their model, which features one-to-one communications and no drift.

Another related self-organization problem is {\em clock synchronization}, where the goal is to maintain a common notion of time in the absence of a global time source. The main difficulty is handling the fact that clocks count time at slightly different rates and that  messages arrive with some delays. Variants of this problem were studied in wireless network contexts, mostly by control theoreticians \cite{wireless,sundararaman2005clock,sivrikaya2004time}. A common technique uses oscillating models \cite{PDE,pulse}. 

\section{Warm-up proofs} \label{sec:warm-up_proofs}

The goal of this section is to prove Theorems~\ref{thm:weighted_average1} and~\ref{thm:weighted_average2}.
These theorems assume all agents execute a weighted-average algorithm with the same responsiveness parameter~$\rho$, that furthermore remains fixed throughout the execution (see Eq.~\eqref{eq:responsive}).
The proofs of all intermediate results can be found in Appendix~\ref{app:weighted}.
We start with the following observation.

\begin{lemma} \label{lem:tb_sum}
    For any integer $t$, and whatever the positions of the agents, we have
    $ \sum_{i\in I} \tb_i^{(t)} = 0$. 
\end{lemma}

Next, we compute how the stretch  $\tb_i^{(t)}$ of Agent $i$ changes when all agents perform weighted-average moves with constant responsiveness parameter~$\rho$.

\begin{lemma} \label{lem:new_theta_bar}
Assume that all agents execute $\W(\rho)$, for some $0\leq \rho\leq 1$. Let $\calE_j^{(t)} = \rho N_{m,j}^{(t)} + N_{d,j}^{(t)}$.
Then for every $i \in I$ and every $t \in \mathbb{N}$,
    \begin{equation*}
        \tb_i^{(t+1)} = \left( 1 - \tfrac{n}{n-1} \rho \right) \tb_i^{(t)} + \tfrac{1}{n-1} \sum_{j\in I\setminus i} \calE_j^{(t)} - \calE_i^{(t)}.
    \end{equation*}
\end{lemma}

Now we can prove that the stretch of each agent is normally distributed at every round, and compute its variance.

\begin{lemma} \label{lem:theta_bar}
    Assume that all agents execute $\W(\rho)$, for some $0\leq \rho\leq 1$. Then, for every $i \in I$ and every $t \in \mathbb{N}$, the stretch $\tb_i^{(t)}$ is normally distributed, $\bbE \pa{\tb_i^{(t)}} = 0$, and
    \begin{equation*}
        \Var\pa{\tb_i^{(t+1)}} = (1-\tfrac{n \rho}{n-1})^2 \cdot \Var\pa{\tb_i^{(t)}} + \frac{n(\rho^2 \smeas^2 + \smove^2)}{n-1}.
    \end{equation*}
\end{lemma}

Before proving Theorem~\ref{thm:weighted_average1}, we need a small technical result.

\begin{claim} \label{claim:convergence}
    Let $a,b\geq0$. Consider the sequence $\{u_n\}_{n=0}^\infty$ defined by letting $u_0 \in \bbR$ and for every integer~$n$, $u_{n+1} = a u_n + b$.  If $a<1$, then $\{u_n\}_{n=0}^\infty$ converges and $\lim_{n \rightarrow +\infty} u_n = b/(1-a)$. If $a = 1$, and~$b>0$, then $\lim_{n \rightarrow +\infty} u_n = +\infty$.
\end{claim}

\begin{proof}(of Theorem~\ref{thm:weighted_average1}.)
   First,  apply Lemma~\ref{lem:theta_bar}. Hence, by Claim~\ref{claim:convergence}, $\Var \big(\tb_i^{(t)}\big)$ converges to the limit as stated.
    Note that the variance is infinite if either (1) the responsiveness is equal to~$0$, in which case the drift adds up endlessly, or (2) the responsiveness is equal to~$1$ and $n=2$, in which case the agents ``swap'' at each round, producing the same result. This completes the proof of Theorem~\ref{thm:weighted_average1}. 
\end{proof}

\begin{proof}(of Theorem \ref{thm:weighted_average2}.)
    Consider the function
    \begin{equation*}
        \Var(\rho) = \frac{ \tfrac{n}{n-1}(\rho^2\smeas^2 + \smove^2)}{1-(1-\tfrac{n}{n-1} \rho)^2}=\frac{\rho^2\smeas^2 + \smove^2}{2\rho - \tfrac{n}{n-1}\rho^2}.
    \end{equation*}
    Note that this function evaluates to~$+\infty$ when~$\rho=0$, or when~$\rho=1$ and~$n=2$. Computing its derivative yields:
    \begin{equation*}
        \Var'(\rho) = \frac{2\rho\smeas^2 \pa{2\rho - \tfrac{n}{n-1}\rho^2} - (\rho^2\smeas^2 + \smove^2)\pa{2-2\tfrac{n}{n-1}\rho}}{\pa{2\rho - \tfrac{n}{n-1}\rho^2}^2}.
    \end{equation*}
Hence, $\Var'(\rho) = 0$ if and only if:
\begin{equation*}
    \rho^2\smeas^2 \pa{2 - \tfrac{n}{n-1}\rho} 
        - (\rho^2\smeas^2 + \smove^2)\pa{1-\tfrac{n}{n-1}\rho} = 0 \iff \rho^2\smeas^2 - \smove^2\pa{1-\tfrac{n}{n-1}\rho} = 0.
    \end{equation*}
    This equation has a unique solution in the interval~$[0,1]$, which is $\rho = \tfrac{ \smove \sqrt{4\smeas^2 + \pa{\tfrac{n}{n-1}\smove}^2} - \tfrac{n}{n-1}\smove^2}{2\smeas^2}$. Checking that this is a minimum  concludes the proof. 
\end{proof}

\section{Solving the alignment problem} \label{sec:flocking_problem}

We denote $\calN \pa{\mu,\Sigma}$ the multivariate normal distribution with mean vector $\mu \in \bbR^{n}$ and co-variance matrix $\Sigma \in \bbR^{n \times n}$, and by $I$ the identity matrix.
Letting $\one$ be the matrix whose all coefficients are equal to $1$, we denote
\[\calM(a,b) = b\one + (a-b)I,\]
the matrix having diagonal coefficients equal to $a$, and other coefficients equal to $b$.

\subsection{Rephrasing the alignment problem as a linear filtering problem}

The goal of this section is to write Eqs.~\eqref{eq:noise} and~\eqref{eq:update_rule} in a matrical form.
Let~$\tb^{(t)}$, $d\theta^{(t)}$, $Y^{(t)}$, $N_m^{(t)}$, and~$N_d^{(t)}$ denote the vectors (indexed by the set of agents) corresponding to each variable. (E.g., $\tb^{(t)} = (\tb_1^{(t)},\ldots,\tb_n^{(t)})$, and similarly for the other vectors.)

\paragraph{\bf Measurement rule} Recall that, as stated in Eq.~\eqref{eq:noise}, the measurement of Agent~$i$ at time $t$ is given by 
$Y_i^{(t)} = \tb_i^{(t)} + N_{m,i}^{(t)}$. 
We simply rewrite this equation as:
\begin{equation} \label{eq:kalman_meas2}
    Y^{(t)} = \tb^{(t)} + N_m^{(t)},
\end{equation}
where, by definition, $N_m^{(t)} \sim \calN \pa{0,\smeas^2 I}$. 

\paragraph{\bf Update rule} We recall the update equation of the stretch, which follows from Eq.~\eqref{eq:update_rule} (see also Eq.~\eqref{eq:new_theta_bar}):
\begin{equation} \label{eq:stretch_update}
\begin{split}
    &\tb_i^{(t+1)} = \tb_i^{(t)} - d\theta_i^{(t)} - N_{d,i}^{(t)} + \pa{ \centre{i}{t+1} - \centre{i}{t} } \\
    &= \tb_i^{(t)} - d\theta_i^{(t)} - N_{d,i}^{(t)} + \tfrac{1}{n-1} \sum_{j\in I\setminus i} \pa{  d\theta_j^{(t)} + N_{d,j}^{(t)} }.
\end{split}
\end{equation}
We define the matrix~\[M_n = \calM \pa{ -1, \tfrac{1}{n-1} }.\]
Let $\tilde{N}_d^{(t)} = M_n N_d^{(t)}$. With these definitions and Eq.~\eqref{eq:stretch_update}, we write Eq.~\eqref{eq:update_rule} in the following vector notation:
\begin{equation} \label{eq:kalman_update2}
    \tb^{(t+1)} = \tb^{(t)} + M_n d\theta^{(t)} + \tilde{N}_d^{(t)}.
\end{equation}
By definition, $N_d^{(t)} \sim \calN \pa{0,\smove^2 I}$, so by Claim~\ref{claim:multivariate_affine_transf}, $\tilde{N}_d^{(t)} \sim \calN \pa{0,Q}$ where~$Q = \smove^2 \cdot M_n I M_n^\top = \smove^2 M_n^2$.

\subsection{Applying the Kalman filter}

We use the Kalman filter as a tool for analysis purposes of the process. We stress, however, that the Kalman filter is not actually applied by any agent. 

We denote by $\thh_t$ the estimate of $\theta^{(t)}$ produced by the Kalman filter \textit{after} the measurements at round $t$ were obtained (Eq.~\eqref{eq:kalman_meas2}) and before the update of round $t$ occurs (Eq.~\eqref{eq:kalman_update2}). Let~$P_t$ denote the error co-variance matrix associated with this estimate. Specifically, $P_t = \bbE \pa{ (\theta^{(t)} - \thh_t) (\theta^{(t)} - \thh_t)^\top }$.
We add the superscript~``$\minus$'' to these notations (for example $P_t^\minus$), to denote the same quantities \textit{before} the measurement obtained at round~$t$.

So, in a sense, round $t$ can be divided into the following four consecutive time steps: (a) the filter produces an estimation $\thh_t^\minus$ of $\theta^{(t)}$, (b) a new measurement vector~$Y^{(t)}$ of $\theta^{(t)}$ is obtained, (c) the filter produces an estimation $\thh_t$ of $\theta^{(t)}$ given the new measurement, and (d) $\theta^{(t)}$ is updated to~$\theta^{(t+1)}$.
Recall that at round 0, i.e., at the initialization stage, the agents are normally distributed around 0. For technical reasons, we define the Kalman filter estimate at round 0 to be zero, that is, $\thh_0^\minus = 0$.

\paragraph{\bf Measurement update} In order to produce (c), the filter incorporates the measurement in (b) to the estimation in~(a). Specifically, the filter computes the ``Kalman gain'':
\begin{equation} \label{eq:kalman_def_1}
    K_t = P_t^\minus \pa{P_t^\minus + \smeas^2 I}^{-1}.
\end{equation}
Then, it produces the following estimate, as required in  (c):
\begin{equation} \label{eq:kalman_def_2}
    \thh_t = \thh_t^\minus + K_t \pa{Y^{(t)} - \thh_t^\minus}.
\end{equation}
The new error co-variance matrix can then be written as:
\begin{equation} \label{eq:kalman_def_3}
    P_t = \pa{I-K_t} P_t^\minus.
\end{equation}

\paragraph{\bf Time update} The estimation for (a) for the following round is then given by:
\begin{equation} \label{eq:kalman_def_4}
    \thh_{t+1}^\minus = \thh_t + M_n d\theta^{(t)}.
\end{equation}
Finally, the new error co-variance matrix that is used in the Kalman gain corresponding to round $t+1$ (Eq.~\eqref{eq:kalman_def_1}) is:
\begin{equation} \label{eq:kalman_def_5}
    P_{t+1}^\minus = P_t + \smove^2 \cdot M_n^2.
\end{equation}

\subsection{Optimality}

\begin{definition}
    A {\em history}~$\calH_t$ is a realization of all measurements and all moves up to time~$(t-1)$. 
\end{definition}
The following is a well-known result, see e.g., \cite{anderson2012optimal}. 
\begin{theorem} \label{thm:kalman_filter_opt}
    For every round~$t$, every history~$\calH_{t}$ and every vector~$x \in \bbR^n$,
    \begin{equation*}
        \sum_{i\in I} \bbE \pa{(\theta_i^{(t)} - \thh_{t,i}^\minus)^2 \mid \calH_{t}} \leq \sum_{i\in I} \bbE \pa{(\theta_i^{(t)} - x_i)^2 \mid \calH_{t}}. 
    \end{equation*}
    In particular, for every round~$t$, every history~$\calH_{t}$, every~$i \leq n$ and every scalar~$x \in \bbR$,
    \begin{equation*}
        \bbE \pa{(\theta_i^{(t)} - \thh_{t,i}^\minus)^2 \mid \calH_{t}} \leq \bbE \pa{(\theta_i^{(t)} - x)^2 \mid \calH_{t}}. 
    \end{equation*}
\end{theorem}

\begin{definition}
    We say that an algorithm for the alignment problem is {\em Kalman-perfect} if it always produces a sequence of moves $(d\theta^{(t)})_{t \geq 0}$, such that for every integer~$t\geq 1$,
    $\thh_{t}^\minus = 0$. 
\end{definition}
The proof of the following result is provided in Appendix~\ref{app:proof_prop_2}.
\begin{proposition} \label{thm:flocking_opt}
    If there exists a Kalman-perfect algorithm for the alignment problem, then this algorithm is optimal in the centralized setting (in the sense of Definition~\ref{def:algorithm_optimality}). Moreover, any other optimal (deterministic) algorithm is Kalman-perfect.
\end{proposition}

\subsection{Kalman gain and the error co-variance matrix}

Recall the sequences $(\alpha_t)_{t \geq 0}$ and $(\rho_\star^{(t)})_{t \geq 0}$  introduced in Definitions~\ref{def:alpha} and~\ref{def:rho_star_t}.
This section is dedicated to proving the following lemma, which provides a formula for the error co-variance matrix associated with the estimate~$\thh_t^\minus$, as well as for the Kalman gain.

\begin{lemma} \label{thm:kalman_filter}
    For every integer $t$, $P_t^\minus = \calM \pa{ \alpha_t, \tfrac{-\alpha_t}{n-1} } = -\alpha_t M_n$, and $K_t = \frac{-\tfrac{n-1}{n} \alpha_t M_n}{\tfrac{n}{n-1}\smeas^2 + \alpha_t} = -\rho_\star^{(t)} M_n.$
\end{lemma}
\begin{proof}
    We prove the first part of the claim by induction, and prove that for every round $t$, the second part of the claim (regarding $K_t$) follows from the first part (regarding $P^\minus_t$).
    
    By Claim~\ref{claim:multivariate_affine_transf} (see Appendix~\ref{sec:preliminary}), $P_0^\minus = \sigma_0^2 M_n^2$.
    By Claim~\ref{claim:multiplication} (see Appendix~\ref{sec:preliminary}),
    \begin{equation*}
        M_n^2 = \calM \pa{ \tfrac{n}{n-1} , - \tfrac{n}{(n-1)^2} } = \tfrac{-n}{n-1} M_n.
    \end{equation*}
    Therefore, $P_0^\minus = -\tfrac{n}{n-1} \sigma_0^2 M_n$, and so the first part of the claim holds at round $0$ since $\alpha_0 = \tfrac{n}{n-1} \sigma_0^2$. 
    
    Now, let us assume that the first part of the claim holds for some $t \in \bbN$. It follows that $P_t^\minus + \smeas^2 I = \calM \pa{ \alpha_t + \smeas^2, \tfrac{-\alpha_t}{n-1} }$.
    By Claim~\ref{claim:inverse} (see Appendix~\ref{sec:preliminary}), since $\smeas > 0$, we have
    \begin{align*}
        \pa{ P_t^\minus + \smeas^2 I }^{-1} = \frac{\calM \pa{ \alpha_t + \smeas^2 - (n-2) \cdot \tfrac{\alpha_t}{n-1} , \tfrac{\alpha_t}{n-1} }}{ \pa{ \alpha_t + \smeas^2 + \tfrac{\alpha_t}{n-1} } \pa{ \alpha_t + \smeas^2 - (n-1) \cdot \tfrac{\alpha_t}{n-1} }} = \frac{\calM \pa{ \tfrac{\alpha_t}{n-1} + \smeas^2 , \tfrac{\alpha_t}{n-1} }}{ \smeas^2 \pa{ \tfrac{n}{n-1}\alpha_t + \smeas^2 }}.
    \end{align*}
    By Claim~\ref{claim:multiplication} again, we can compute the ``Kalman gain'':
    \begin{align*}
        K_t &= P_t^\minus \pa{ P_t^\minus + \smeas^2 I }^{-1} \\
        &= \frac{1} { \smeas^2 \pa{ \tfrac{n}{n-1}\alpha_t + \smeas^2 }} \cdot \calM \bigg( \tfrac{\alpha_t^2}{n-1} + \smeas^2 \alpha_t - \tfrac{(n-1)\cdot \alpha_t^2}{(n-1)^2} ,  \tfrac{-\alpha_t^2}{(n-1)^2} - \tfrac{\smeas^2 \alpha_t}{n-1} + \tfrac{\alpha_t^2}{n-1} - (n-2) \cdot \tfrac{\alpha_t^2}{(n-1)^2} \bigg) \\
        &= \frac{\calM \pa{ \smeas^2 \alpha_t , \tfrac{-\smeas^2 \alpha_t}{n-1} }}{ \smeas^2 \pa{ \tfrac{n}{n-1}\alpha_t + \smeas^2 }} 
        = \frac{-\alpha_t M_n}{\tfrac{n}{n-1}\alpha_t + \smeas^2},
    \end{align*}
    This proves that the second part of the claim holds at round $t$.
    Next, we compute the error co-variance matrix after the measurement:
    \begin{align*}
        P_t &= (I - K_t) \cdot  P_t^\minus = \frac{\calM \pa{\tfrac{\alpha_t}{n-1} + \smeas^2 , \tfrac{\alpha_t}{n-1}}}{\tfrac{n}{n-1}\alpha_t + \smeas^2} \cdot P_t^\minus \\
        &= \frac{1} {\tfrac{n}{n-1}\alpha_t + \smeas^2} \cdot \calM \bigg( \tfrac{\alpha_t^2}{n-1} + \smeas^2 \alpha_t - (n-1) \cdot \tfrac{\alpha_t^2}{(n-1)^2} , -\tfrac{\alpha_t^2}{(n-1)^2} - \tfrac{\smeas^2 \alpha_t}{n-1} + \tfrac{\alpha_t^2}{n-1} - (n-2) \cdot \tfrac{\alpha_t^2}{(n-1)^2} \bigg) \\
        &= \frac{\calM \pa{\smeas^2 \alpha_t ,  \tfrac{-\smeas^2 \alpha_t}{n-1} }}{\tfrac{n}{n-1}\alpha_t + \smeas^2} 
        = -\frac{\tfrac{n-1}{n} \smeas^2 \alpha_t}{\alpha_t + \tfrac{n-1}{n} \smeas^2} M_n.
    \end{align*}
    Eventually, we compute the error co-variance matrix at round~$t+1$, before the measurement:
    \begin{equation*}
        P_{t+1}^\minus = P_t + \smove^2 M_n^2 = P_t - \tfrac{n}{n-1}\smove^2 M_n.
    \end{equation*}
   Plugging in the expression of $P_t$, we get
    \begin{align*}
        P_{t+1}^\minus = -\frac{\tfrac{n-1}{n} \smeas^2 \alpha_t}{\alpha_t + \tfrac{n-1}{n} \smeas^2} M_n - \tfrac{n}{n-1}\smove^2 M_n = - \pa{\frac{\tfrac{n-1}{n} \smeas^2 \alpha_t}{\tfrac{n-1}{n} \smeas^2 + \alpha_t} + \tfrac{n}{n-1}\smove^2} M_n 
        = -\alpha_{t+1} M_n,
    \end{align*}  
    which concludes the induction proof.
\end{proof}

\subsection{Meet-at-the-center}\label{sec:meet}

The goal of this section is to prove Theorem \ref{thm:matc_opt} stating that Algorithm Meet-at-the-center is optimal in the centralized setting (in the sense of Definition \ref{def:algorithm_optimality}). 
The choice of the name comes from the fact that this algorithm minimizes the drift of the center of mass (see Theorem~\ref{thm:drift_matc}).

\begin{proof}(of Theorem \ref{thm:matc_opt}.)
    By Proposition~\ref{thm:flocking_opt}, it is sufficient to prove that Algorithm Meet-at-the-center is Kalman-perfect, that is, that  for every integer~$t\geq 1$, the Kalman filter associated with the moves it produces gives the estimate $\thh_t^\minus = 0$. 
    
    For this purpose, assume that all agents run Algorithm Meet-at-the-center. We prove by induction that for every integer~$t\geq 0$, the Kalman filter produces the estimate $\thh_t^\minus = 0$.
    This holds at $t=0$ since we assumed that the Kalman filter estimates 0 at round 0, i.e.,~$\thh_0^\minus=0$.
    Next, let us assume that $\thh_t^\minus = 0$ for some integer~$t\geq 0$. We have by definition,
   \begin{equation} \label{eq:def_recall1'}
        \thh_{t+1}^\minus = \thh_t + M_n d\theta^{(t)},
    \end{equation}
    and
    \begin{equation} \label{eq:def_recall2'}
        \thh_t = \thh_t^\minus + K_t \pa{Y^{(t)} - \thh_t^\minus} = K_t Y^{(t)},
    \end{equation}
    where the second equality in Eq.~\eqref{eq:def_recall2'} is by the  induction hypothesis. 
    By Lemma~\ref{thm:kalman_filter},
    $K_t = -\rho_\star^{(t)} M_n.$
    Note that, by definition of Algorithm Meet-at-the-center,
    $d\theta^{(t)} = -\tfrac{n-1}{n}\rho_\star^{(t)} M_n Y^{(t)} = \tfrac{n-1}{n} \thh_t.$
    Finally, by the equations above, we rewrite Eq.~\eqref{eq:def_recall1'} as:
    \begin{align*}
        \thh_{t+1}^\minus = \pa{I + \tfrac{n-1}{n} M_n} \thh_t = \pa{I + \tfrac{n-1}{n} M_n} K_t Y^{(t)} = - \pa{M_n + \tfrac{n-1}{n} M_n^2} \rho_\star^{(t)} Y^{(t)}.
    \end{align*}
    Since $M_n^2 = \tfrac{-n}{n-1} M_n$, we have~$\thh_{t+1}^\minus = 0$. This concludes the induction, and completes the proof of the theorem.
\end{proof}

\subsection{Algorithm \texorpdfstring{$\W^{\star}$}{Wstar} is optimal in the centralized setting}\label{sec:optW} 
The goal of this section is to prove Theorem~\ref{thm:weight_opt}. The proof follows the same line of arguments as the proof of Theorem~\ref{thm:matc_opt}.

\begin{proof}(of Theorem~\ref{thm:weight_opt})
    Our goal is to prove that Algorithm~$\W^{\star}$ is Kalman-perfect, that is, for every integer~$t\geq 1$, the Kalman filter associated with the moves produced by Algorithm~$\W^{\star}$ gives the estimate $\thh_t^\minus = 0$. This would conclude the proof of the theorem, by  Proposition~\ref{thm:flocking_opt}. 
    
    For this purpose, assume that all agents run Algorithm~$\W^{\star}$. We prove by induction that for every integer~$t\geq 0$, the Kalman filter produces the estimate $\thh_t^\minus = 0$.

    The base case, where $t=0$, holds since we assumed that the Kalman filter estimates zero at round zero, i.e.,~$\thh_0^\minus=0$.
    Next, let us assume that $\thh_t^\minus = 0$ for some integer~$t\geq 0$, and consider $t+1$. We have by definition,
    \begin{equation} \label{eq:def_recall1}
        \thh_{t+1}^\minus = \thh_t + M_n d\theta^{(t)},
    \end{equation}
    and
    \begin{equation} \label{eq:def_recall2}
        \thh_t = \thh_t^\minus + K_t \pa{Y^{(t)} - \thh_t^\minus} = K_t Y^{(t)},
    \end{equation}
    where the second equality in Eq.~\eqref{eq:def_recall2} is by induction hypothesis. By Lemma~\ref{thm:kalman_filter}, $K_t = -\rho_\star^{(t)} M_n$. Hence, Eq.~\eqref{eq:def_recall1} rewrites
    \begin{equation*}
        \thh_{t+1}^\minus = M_n \pa{ -\rho_\star^{(t)} Y^{(t)} + d\theta^{(t)} }.
    \end{equation*}
    Finally, by the definition of~$\W^{\star} $, we have $d\theta^{(t)} = \rho_\star^{(t)} Y^{(t)}$, so $\thh_{t+1}^\minus = 0$, concluding the induction step.
\end{proof}

\section{Discussion and Future Work}

We introduced the distributed {\em alignment} problem, aiming to study basic algorithmic aspects that govern approximate agreement processes in unreliable stochastic environments, while relying on measurements of the overall group tendency. 
We proposed a weighted-average algorithm, termed $\W^\star$, and proved that it achieves optimal stretches, that is, it minimizes the expected distance of each agent from the center of mass of others, at each round.
In fact, the optimality of $\W^\star$ holds also in the full-communication (or centralized) model, where communication of internal states is unrestricted. 
Algorithm $\W^{\star}$ is based on having a different responsiveness parameter that is applied by each agent at each round. We also find the optimal weighted-average algorithm among the ones that use the same responsiveness at all times. In contrast to $\W^{\star}$ that minimizes the expected stretch at every round separately, this algorithm becomes optimal only at steady-state. 

In our setting, the cost of an individual is defined as its expected distance from the average position of other agents.
Minimizing this quantity is equivalent to minimizing the expected distance from the average position of all agents (see Footnote 2).
Another interesting measure is the expected {\em diameter} of the group, defined as the maximal distance between any two agents, at steady state. It would not come as a surprise if Algorithm~$\W^\star$ would turn out to be optimal also with respect to this measure, however, analyzing its expected diameter would require handling further dependencies between agents, and therefore remains for future work.  

It would also be interesting to see how our work could be extended to scenarios in which each agent can observe a subset of randomly chosen agents at each round. As mentioned in the introduction, averaging the values of the agents in such a subset can serve as a noisy sample of average value in the population. However, in contrast to our model, the noise in this sample is not fixed, and depends on the distribution of values in the population and on the size of the sampled subset. In many cases, one can expect that sampling a subset containing a logarithmic number of agents could suffice to obtain a relatively good sample of the population. We find this encouraging towards generalizing our work to such scenarios.

\paragraph{Acknowledgment}
The authors thank Yongcan Cao for a helpful discussion regarding related works on distributed Kalman filter, as well as anonymous reviewers whose critical comments helped improve the presentation of the paper significantly.

\bibliography{ref}{}

\begin{thebibliography}{10}

\bibitem{anderson2012optimal}
Brian~DO Anderson and John~B Moore.
\newblock {\em Optimal filtering}.
\newblock Courier Corporation, 2012.

\bibitem{becchetti2020consensus}
Luca Becchetti, Andrea Clementi, and Emanuele Natale.
\newblock Consensus dynamics: An overview.
\newblock {\em ACM SIGACT News}, 51(1):58--104, 2020.

\bibitem{becchetti2017simple}
Luca Becchetti, Andrea Clementi, Emanuele Natale, Francesco Pasquale, Riccardo
  Silvestri, and Luca Trevisan.
\newblock Simple dynamics for plurality consensus.
\newblock {\em Distributed Computing}, 30(4):293--306, 2017.

\bibitem{berman2011study}
Spring Berman, Quentin Lindsey, Mahmut~Selman Sakar, Vijay Kumar, and Stephen
  Pratt.
\newblock Study of group food retrieval by ants as a model for multi-robot
  collective transport strategies.
\newblock {\em Robotics: Sci. Syst. VI}, page 259, 2011.

\bibitem{amos-soda}
Lucas Boczkowski, Amos Korman, and Emanuele Natale.
\newblock Minimizing message size in stochastic communication patterns: fast
  self-stabilizing protocols with 3 bits.
\newblock {\em Distributed Comput.}, 32(3):173--191, 2019.

\bibitem{cao2012overview}
Yongcan Cao, Wenwu Yu, Wei Ren, and Guanrong Chen.
\newblock An overview of recent progress in the study of distributed
  multi-agent coordination.
\newblock {\em IEEE Transactions on Industrial informatics}, 9(1):427--438,
  2012.

\bibitem{conradt2005consensus}
Larissa Conradt and Timothy~J Roper.
\newblock Consensus decision making in animals.
\newblock {\em Trends in ecology \& evolution}, 20(8):449--456, 2005.

\bibitem{doerr2011stabilizing}
Benjamin Doerr, Leslie~Ann Goldberg, Lorenz Minder, Thomas Sauerwald, and
  Christian Scheideler.
\newblock Stabilizing consensus with the power of two choices.
\newblock In {\em Proceedings of the twenty-third annual ACM symposium on
  Parallelism in algorithms and architectures}, pages 149--158, 2011.

\bibitem{feinerman2017breathe}
Ofer Feinerman, Bernhard Haeupler, and Amos Korman.
\newblock Breathe before speaking: efficient information dissemination despite
  noisy, limited and anonymous communication.
\newblock {\em Distributed Computing}, 30(5):339--355, 2017.

\bibitem{ncomm-ofer}
Aviram Gelblum, Itai Pinkoviezky, Ehud Fonio, Abhijit Ghosh, Nir Gov, and Ofer
  Feinerman.
\newblock Ant groups optimally amplify the effect of transiently informed
  individuals.
\newblock {\em Nature communications}, 6:7729, 2015.

\bibitem{gelblum2016emergent}
Aviram Gelblum, Itai Pinkoviezky, Ehud Fonio, Nir~S Gov, and Ofer Feinerman.
\newblock Emergent oscillations assist obstacle negotiation during ant
  cooperative transport.
\newblock {\em Proceedings of the National Academy of Sciences},
  113(51):14615--14620, 2016.

\bibitem{ghaffari2018nearly}
Mohsen Ghaffari and Johannes Lengler.
\newblock Nearly-tight analysis for 2-choice and 3-majority consensus dynamics.
\newblock In {\em Proceedings of the 2018 ACM Symposium on Principles of
  Distributed Computing}, pages 305--313, 2018.

\bibitem{gorbonos2016long}
Dan Gorbonos, Reuven Ianconescu, James~G Puckett, Rui Ni, Nicholas~T Ouellette,
  and Nir~S Gov.
\newblock Long-range acoustic interactions in insect swarms: an adaptive
  gravity model.
\newblock {\em New Journal of Physics}, 18(7):073042, 2016.

\bibitem{guerraoui2015byzantine}
Rachid Guerraoui and Alexandre Maurer.
\newblock Byzantine fireflies.
\newblock In {\em International Symposium on Distributed Computing}, pages
  47--59. Springer, 2015.

\bibitem{pulse}
Yao-Win Hong and Anna Scaglione.
\newblock A scalable synchronization protocol for large scale sensor networks
  and its applications.
\newblock {\em IEEE Journal on Selected Areas in Communications},
  23(5):1085--1099, 2005.

\bibitem{huang2009coordination}
Minyi Huang and Jonathan~H Manton.
\newblock Coordination and consensus of networked agents with noisy
  measurements: stochastic algorithms and asymptotic behavior.
\newblock {\em SIAM Journal on Control and Optimization}, 48(1):134--161, 2009.

\bibitem{korman2021sequential}
Amos Korman, Oran Ayalon, Yigal Sternklar, Ehud Fonio, Nir Gov, and Ofer
  Feinerman.
\newblock Sequential decision-making in ants and implications to the evidence
  accumulation decision model.
\newblock {\em Frontiers in Applied Mathematics and Statistics}, 7:37, 2021.

\bibitem{fisher}
Amos Korman, Efrat Greenwald, and Ofer Feinerman.
\newblock Confidence sharing: An economic strategy for efficient information
  flows in animal groups.
\newblock {\em PLoS Computational Biology}, 10(10):e1003862, 2014.

\bibitem{korman:hal-03124213}
Amos Korman and Robin Vacus.
\newblock {Distributed Alignment Processes with Samples of Group Average}.
\newblock https://hal.archives-ouvertes.fr/hal-03124213, February 2022.

\bibitem{mccreery2014cooperative}
Helen~F McCreery and MD~Breed.
\newblock Cooperative transport in ants: a review of proximate mechanisms.
\newblock {\em Insectes sociaux}, 61(2):99--110, 2014.

\bibitem{PDE}
Renato~E Mirollo and Steven~H Strogatz.
\newblock Synchronization of pulse-coupled biological oscillators.
\newblock {\em SIAM Journal on Applied Mathematics}, 50(6):1645--1662, 1990.

\bibitem{olfati2007distributed}
Reza Olfati-Saber.
\newblock Distributed kalman filtering for sensor networks.
\newblock In {\em 2007 46th IEEE Conference on Decision and Control}, pages
  5492--5498. IEEE, 2007.

\bibitem{ren2005survey}
Wei Ren, Randal~W Beard, and Ella~M Atkins.
\newblock A survey of consensus problems in multi-agent coordination.
\newblock In {\em Proceedings of the 2005, American Control Conference, 2005.},
  pages 1859--1864. IEEE, 2005.

\bibitem{reynolds}
Craig~W Reynolds.
\newblock {\em Flocks, herds and schools: A distributed behavioral model},
  volume~21.
\newblock ACM, 1987.

\bibitem{wireless}
Osvaldo Simeone, Umberto Spagnolini, Yeheskel Bar-Ness, and Steven~H Strogatz.
\newblock Distributed synchronization in wireless networks.
\newblock {\em IEEE Signal Processing Magazine}, 25(5):81--97, 2008.

\bibitem{simons2004many}
Andrew~M Simons.
\newblock Many wrongs: the advantage of group navigation.
\newblock {\em Trends in ecology \& evolution}, 19(9):453--455, 2004.

\bibitem{sivrikaya2004time}
Fikret Sivrikaya and B{\"u}lent Yener.
\newblock Time synchronization in sensor networks: a survey.
\newblock {\em IEEE network}, 18(4):45--50, 2004.

\bibitem{sundararaman2005clock}
Bharath Sundararaman, Ugo Buy, and Ajay~D Kshemkalyani.
\newblock Clock synchronization for wireless sensor networks: a survey.
\newblock {\em Ad hoc networks}, 3(3):281--323, 2005.

\bibitem{tamm1980bird}
Staffan Tamm.
\newblock Bird orientation: single homing pigeons compared with small flocks.
\newblock {\em Behavioral Ecology and Sociobiology}, 7(4):319--322, 1980.

\bibitem{tang2015continuous}
Huaibin Tang and Tao Li.
\newblock Continuous-time stochastic consensus: Stochastic approximation and
  kalman--bucy filtering based protocols.
\newblock {\em Automatica}, 61:146--155, 2015.

\bibitem{vicsek_review}
Tamás Vicsek and Anna Zafeiris.
\newblock Collective motion.
\newblock {\em Physics Reports}, 517(3):71 -- 140, 2012.
\newblock Collective motion.

\bibitem{wang2016multi}
Zijian Wang and Mac Schwager.
\newblock Multi-robot manipulation without communication.
\newblock In {\em Distributed autonomous robotic systems}, pages 135--149.
  Springer, 2016.

\end{thebibliography}
\bibliographystyle{plain}

\newpage
\appendix

\section{Proof of Proposition~\ref{thm:flocking_opt}}
\label{app:proof_prop_2}

Proposition~\ref{thm:flocking_opt} is a direct consequence of the following result.
\begin{lemma}
    Fix $i \in I$ and an integer~$t\geq 1$, and consider
    a history~$\calH_t$ such that the Kalman filter produces~$\thh_{t,i}^\minus = 0$ on~$\calH_t$.
    Let~$d\vartheta^{(t-1)}$ be an alternative move for the agents at round~$(t-1)$, and consider history~$\calH_t'$ obtained from~$\calH_t$ by replacing~$d\theta^{(t-1)}$ by $d\vartheta^{(t-1)}$.
    Then,
    \begin{equation*}
        \bbE \pa{ {\theta_i^{(t)}}^2 \mid \calH_t' } \geq \bbE \pa{ {\theta_i^{(t)}}^2 \mid \calH_t }.
    \end{equation*}
\end{lemma}

\begin{proof}
    Let $c$ be the~$i$th coordinate of the vector~$M_n (d\theta^{(t)} - d\vartheta^{(t)})$. By Eq.~\eqref{eq:update_rule}, $c$ is the difference between the position of Agent~$i$ in the beginning of round~$t$ if it moves by~$d\theta^{(t)}$ instead of moving by~$d\vartheta^{(t)}$ (conditioning on having the same drift at the end of round $t-1)$.
    In other words,
    \begin{equation} \label{eq:c1}
        \bbE \pa{ {\theta_i^{(t)}}^2 \mid \calH_t' } = \bbE \pa{ (\theta_i^{(t)}-c)^2 \mid \calH_t }.
    \end{equation}
    By Theorem~\ref{thm:kalman_filter_opt},
    \begin{equation} \label{eq:c2}
        \bbE \pa{ (\theta_i^{(t)}-c)^2 \mid \calH_t } \geq \bbE \pa{ (\theta_i^{(t)}-\thh_{t,i})^2 \mid \calH_t }.
    \end{equation}
    By assumption, $\thh_{t,i} = 0$, so
    \begin{equation} \label{eq:c3}
        \bbE \pa{ (\theta_i^{(t)}-\thh_{t,i})^2 \mid \calH_t } = \bbE \pa{ {\theta_i^{(t)}}^2 \mid \calH_t }.
    \end{equation}
    Combining Eqs.~\eqref{eq:c1}, \eqref{eq:c2} and~\eqref{eq:c3} establishes the result.
\end{proof}

\section{More Results on Weighted-Average Algorithms} \label{app:weighted}

\begin{proof} (of Lemma~\ref{lem:tb_sum}.)
    \begin{align*}
        \sum_{i\in I} \tb_i^{(t)} &= \sum_{i\in I} \left( \centre{i}{t} - \theta_i^{(t)} \right) 
        = \tfrac{1}{n-1} \sum_{i\in I} \sum_{j\in I\setminus i} \theta_i^{(t)} - \sum_{i\in I} \theta_i^{(t)} \\
        &= \tfrac{1}{n-1}\cdot (n-1) \sum_{i\in I} \theta_i^{(t)} - \sum_{i\in I} \theta_i^{(t)} = 0.
    \end{align*}
\end{proof}

\begin{proof} (of Lemma~\ref{lem:new_theta_bar}.)
    Agent $i$'s stretch at round~$t+1$ is:
    \begin{equation} \label{eq:new_theta_bar}
    \begin{split}
        &\tb_i^{(t+1)} = \centre{i}{t+1} - \theta_i^{(t+1)} = \centre{i}{t+1} - (\theta_i^{(t)} + d\theta_i^{(t)} + N_{d,i}^{(t)} ) \\
        &= \left( \centre{i}{t+1} - \centre{i}{t} \right) + \pa{\centre{i}{t} - \theta_i^{(t)}} - d\theta_i^{(t)} - N_{d,i}^{(t)} \\
        &= \left( \centre{i}{t+1} - \centre{i}{t} \right) + \tb_i^{(t)} - d\theta_i^{(t)} - N_{d,i}^{(t)},
    \end{split}
    \end{equation}
    where the first equality is by definition, and the second equality is a consequence of~\eqref{eq:update_rule}.
    Let us break down the first term:
    \begin{align*}
        \centre{i}{t+1} - \centre{i}{t} &= \tfrac{1}{n-1}\sum_{j\in I\setminus i} \pa{ d\theta_j^{(t)} + N_{d,j}^{(t)}} \\
        &= \tfrac{1}{n-1}\sum_{j\in I\setminus i} \pa{ \rho \left(\tb_j^{(t)} + N_{m,j}^{(t)} \right) + N_{d,j}^{(t)}},
    \end{align*}
    where the second equality is because
    $d\theta_j^{(t)} = \rho Y_j^{(t)} = \rho \pa{ \tb_j^{(t)} + N_{m,j}^{(t)} }.$
    By Lemma~\ref{lem:tb_sum}, $\sum_{j\in I\setminus i} \tb_j^{(t)} = - \tb_i^{(t)}$, so we can rewrite the last equation as
    \begin{equation} \label{eq:center_drift2}
    \begin{split}
        \centre{i}{t+1} - \centre{i}{t} &= \tfrac{-\rho}{n-1} \tb_i^{(t)} + \tfrac{1}{n-1} \sum_{j\in I\setminus i} \rho N_{m,j}^{(t)} + N_{d,j}^{(t)} \\
        &= \tfrac{-\rho}{n-1} \tb_i^{(t)} + \tfrac{1}{n-1} \sum_{j\in I\setminus i} \calE_j^{(t)}.
    \end{split}
    \end{equation}
    Plugging Eq.~\eqref{eq:center_drift2} into Eq.~\eqref{eq:new_theta_bar} gives
    \begin{equation} \label{eq:new_theta_bar2}
        \tb_i^{(t+1)} = \pa{1 - \tfrac{\rho}{n-1}} \tb_i^{(t)} + \tfrac{1}{n-1} \sum_{j\in I\setminus i} \calE_j^{(t)} - d\theta_i^{(t)} - N_{d,i}^{(t)}.
    \end{equation}
    By assumption, $d\theta_i^{(t)} = \rho Y_i^{(t)} = \rho \pa{ \tb_i^{(t)} + N_{m,i}^{(t)} }$. Plugging this expression into Eq.~\eqref{eq:new_theta_bar2} gives the result.
\end{proof}

\begin{proof} (of Lemma~\ref{lem:theta_bar}.)
We prove that the stretch is normally distributed with mean~$0$ by induction on~$t$. By construction, for every~$i$, $\tb_i^{(0)}$ is normally distributed with mean~$0$. Let us assume that $\tb_i^{(t)}$ is normally distributed with mean~$0$ for some round~$t$, and consider round $t+1$.
Recall that Lemma~\ref{lem:new_theta_bar} gives
\begin{equation} \label{eq:new_theta_bar3}
    \tb_i^{(t+1)} = \pa{ 1 - \tfrac{n}{n-1}\rho } \tb_i^{(t)} + \tfrac{1}{n-1} \sum_{j \in I\setminus i} \calE_j^{(t)} - \calE_i^{(t)}.
\end{equation}
Since by definition, for every $j$, $\calE_j^{(t)}$ is normally distributed around $0$, and by induction $\tb_i^{(t)}$ is  normally distributed around~$0$, then $\tb_i^{(t+1)}$ is also normally distributed around~$0$. This concludes the induction.
Noting that~$\Var(\calE_j^{(t)}) = \rho^2 \smeas^2 + \smove^2$ and taking the variance in Eq.~\eqref{eq:new_theta_bar3} concludes the proof.
\end{proof}

\begin{proof} (of Claim~\ref{claim:convergence}.)
    First, consider the case that~$a<1$. Let $\lambda = b/(a-1)$. Consider the sequence defined by $v_n = u_n + \lambda$. We have $v_{n+1} = u_{n+1} + \lambda = au_n + b + \lambda 
        = au_n + (a-1)\lambda + \lambda 
        = au_n + a\lambda = a(u_n + \lambda) = av_n$.
    Since $|a| < 1$, $\lim v_n = 0$, and so $\lim u_n = -\lambda = b/(1-a)$.
    
    Now, if~$a=1$, then we have for every~$n \in \bbN$, $u_n = u_0 + nb$. If~$b>0$ then $\lim u_n = +\infty$.
\end{proof}

\section{Useful linear algebra claims} \label{sec:preliminary}

\begin{claim} \label{claim:multivariate_affine_transf}
    If $X \sim \calN \pa{\mu,\Sigma}$, then for every $c \in \bbR^{n}$ and $B \in \bbR^{n \times n}$,
   $c + BX \sim \calN \pa{ c + B\mu , B\Sigma B^\top }$. 
\end{claim}
Claim~\ref{claim:multivariate_affine_transf} above is well-known, and Claim~\ref{claim:multiplication} below follows from a straightforward matrix multiplication.
\begin{claim} \label{claim:multiplication}
    For every $a,b,a',b' \in \bbR$, $\calM(a,b) \calM(a',b') = \calM(aa' + (n-1)bb' , ab' + a'b + (n-2)bb')$.
\end{claim}

\begin{claim} \label{claim:inverse}
    For every $a,b \in \bbR$ such that $a \neq b$ and $a \neq - (n-1)b$, the matrix $\calM(a,b)$ is invertible, and $\calM(a,b)^{-1} = \frac{\calM \pa{ a + (n-2)b , -b }}{(a-b)(a + (n-1)b)}$.
\end{claim}
\begin{proof}
    Note that $\one^2 = n\one$. Let $A = \calM(a,b)$. We have
    \begin{equation*}
    \begin{split}
        A^2 &= (b\one + (a-b)I)^2 
            = b^2 \one^2 + 2b(a-b)\one + (a-b)^2 I \\
            &= (nb + 2(a-b)) (b\one) + (a-b)^2 I \\
            &= (2a + (n-2)b) (b\one + (a-b)I) \\
            &~~~~ - (nb + 2(a-b))(a-b)I + (a-b)^2 I \\
            &= (2a + (n-2)b) A + (a-b)((a-b) - nb - 2(a-b)) I \\
            &= (2a + (n-2)b) A - (a-b)(a + (n-1)b) I.
    \end{split}
    \end{equation*}
    Hence $A (A - (2a + (n-2)b)I) = - (a-b)(a + (n-1)b) I$, from which we conclude the proof of Claim~\ref{claim:inverse}.
\end{proof}

\newcommand{\wstar}[1]{\left[#1 \right]^{\W^\star}}
\newcommand{\matc}[1]{\left[#1 \right]^{\textnormal{MatC}}}
\newcommand{\vtb}{\overline{\vartheta}}

\section{Additional claims about optimal algorithms}

\subsection{All optimal algorithms are shifts of one another} \label{sec:shifts}

In this section, we characterize all optimal (deterministic) algorithms for the alignment problem in the centralized setting. We show that each of these algorithms can be obtained from $\W^\star$, by shifting all the agents by the same quantity $\lambda_t$.
\begin{theorem} \label{thm:opt_charac}
    A deterministic algorithm is optimal in the centralized setting if 
    and only if 
    for every round~$t$, there exists~$\lambda_t$ such that for every~$i \in I$, $d\theta_i^{(t)} = \rho_\star^{(t)} Y_i^{(t)} + \lambda_t$.
\end{theorem}
\begin{proof}
    We have already established that Algorithm $\W^\star$ is Kalman-perfect.
    Therefore, by Proposition~\ref{thm:flocking_opt}, any other deterministic algorithm is optimal in the centralized setting if 
    and only if it is Kalman-perfect. In other words, it is optimal if and only if
    it produces a sequence of moves such that for every round~$t$, the Kalman-filter estimator operating on the corresponding process yields
    \begin{align*}
        \thh_{t+1}^\minus = 0 &\iff \thh_t + M_n d\theta^{(t)} = 0 \\
        &\iff K_t Y^{(t)} + M_n d\theta^{(t)} = 0 \\
        &\iff - \rho_\star^{(t)} M_n Y^{(t)} + M_n d\theta^{(t)} = 0 \\
        &\iff M_n \pa{ - \rho_\star^{(t)} Y^{(t)} + d\theta^{(t)} } = 0 \\
        &\iff - \rho_\star^{(t)} Y^{(t)} + d\theta^{(t)} \in \ker(M_n).
    \end{align*}
    Writing~$\textbf{1}$ to denote the vector whose coefficients are all equal to~$1$, we observe that~$\textbf{1} \in \ker(M_n)$.
    Since $\textnormal{rank}(M_n) = n-1$, $\dim(\ker(M_n)) = 1$, so for every round~$t$,
    \begin{multline*}
        - \rho_\star^{(t)} Y^{(t)} + d\theta^{(t)} \in \ker(M_n) \\
        \iff \exists \lambda_t \in \bbR, ~- \rho_\star^{(t)} Y^{(t)} + d\theta^{(t)} = \lambda_t \cdot \textbf{1} \\
        \iff \exists \lambda_t \in \bbR, ~d\theta^{(t)} = \lambda_t \cdot \textbf{1} + \rho_\star^{(t)} Y^{(t)}.
    \end{multline*}
\end{proof}

\subsection{Proof of Claim~\ref{claim:best_drift}} \label{sec:last_proof}

    Let~$A$ be an optimal algorithm.
    Again, we denote the variables involved in the execution of~$A$ by~$[\cdot]^A$.
    We have shown in Appendix~\ref{sec:shifts} that all optimal algorithms are shifts of one another, i.e., for all~$t$ we can find a random variable~$\lambda_t$ s.t. for all~$i$,
    $\left[ {d\theta_i^{(t)}} \right]^A = \matc{d\theta_i^{(t)}} + \lambda_t$.
    It implies that
    $\tfrac{1}{n} \sum_{i\in I} \left[ d\theta_i^{(t)} \right]^A = \lambda_t + \tfrac{1}{n} \sum_{i\in I} \matc{d\theta_i^{(t)}}$, which is equal to~$\lambda_t$ (the second term is~$0$ following the definition of Algorithm~\ref{alg:MatC}).
    Therefore, $\left[\DRel_t \right]^A = \tfrac{1}{n} \sum_{s=0}^t \sum_{i\in I} \left[ d\theta_i^{(s)} \right]^A = \sum_{s=0}^t \lambda_s$.
    If~$A$ is different than Algorithm~\ref{alg:MatC}, we can define~$t_0 := \min \{ t \in \bbN, \bbE\pa{ \lambda_{t}^2} > 0 \} < +\infty$, and have
    \begin{equation*}
        \bbE \pa{ {\left[\DRel_{t_0} \right]^A}^2} = \bbE \bigg( \bigg( \sum_{s=0}^{t_0} \lambda_s \bigg)^2 \bigg) \geq \sum_{s=0}^{t_0} \bbE\pa{ \lambda_s^2 } >  0.
    \end{equation*}
    By Eq.~\eqref{eq:lb_var}, we conclude that, among optimal algorithm, ``meet at the center'' is the only one minimizing~$\Delta_t$ for every~$t$. Note that it is not implementable in the distributed setting, because it requires for each agent to know all measurements. This concludes the proof of Claim~\ref{claim:best_drift}.

\end{document}